%% file: social_optimum_gt12.tex
\documentclass[10pt,conference,letter]{IEEEtran}
\usepackage[latin1]{inputenc}
\usepackage[cmex10]{amsmath}
\usepackage{amsfonts}
\usepackage{amssymb}
\usepackage{xspace}
\usepackage{color}
\usepackage{graphicx,subfigure}
\usepackage{amsthm}
\usepackage{enumerate}
\usepackage{multirow}
\usepackage{cite}
\usepackage{algpseudocode}
\usepackage{algorithm}
\usepackage{multicol}

\newcommand{\gt}{the \emph{GT criterion}\xspace}
\newtheorem{lemma}{Lemma}
\newtheorem{thm}{Theorem}
\newtheorem{prop}{Proposition}
\newtheorem{cor}{Corollary}
\makeatletter
\newcounter{subalgorithm}
\newenvironment{subalgorithms}{
	\stepcounter{algorithm}
	\setcounter{subalgorithm}{0}
	\let\old@fnum@algorithm=\fnum@algorithm%
	\def\fnum@algorithm{%
		\addtocounter{algorithm}{-1}%
		\stepcounter{subalgorithm}%
		\old@fnum@algorithm}%
}
\makeatother
\makeatletter
\def\ps@headings{%
	\def\@oddhead{\mbox{}\scriptsize\rightmark \hfil \thepage}%
	\def\@evenhead{\scriptsize\thepage \hfil \leftmark\mbox{}}%
	\def\@oddfoot{}%
	\def\@evenfoot{}}
\makeatother
\begin{document}

\title{Social optimum in Social Groups with Give-and-Take criterion}



\author{
	\IEEEauthorblockN{Saurabh Aggarwal, Joy Kuri}\\
	\IEEEauthorblockA{Department of Electronic Systems Engineering,\\ Indian Institute of Science, Bangalore, India,\\
		Email: saggarwal@cedt.iisc.ernet.in, kuri@cedt.iisc.ernet.in}
	\and	
	\IEEEauthorblockN{Rahul Vaze}\\
	\IEEEauthorblockA{School of Technology and Computer Science,\\Tata Institute of Fundamental Research, Mumbai, India
		\\ Email: vaze@tcs.tifr.res.in}
}
\maketitle
\begin{abstract}
	\input{abstract7}
\end{abstract}
\bibliographystyle{IEEEtran}

\input{Introduction11}

\input{system_model10}

\input{heuristics15}

\input{results9}

\input{conclusions10}

\bibliography{info_exchange}
\end{document}

%% file: abstract7.tex
We consider a ``Social Group'' of networked nodes, seeking a ``universe'' of segments. Each node has subset of the universe, and access to an expensive resource for downloading data. Alternatively, nodes can also acquire the universe by exchanging segments among themselves, at low cost, using a local network interface. While local exchanges ensure minimum cost, ``free riders'' in the group can exploit the system. To prohibit free riding, we propose the ``Give-and-Take'' criterion, where exchange is allowed if each node has segments unavailable with the other. Under this criterion, we consider the problem of maximizing the aggregate cardinality of the nodes' segment sets. First, we present a randomized algorithm, whose analysis yields a lower bound on the expected aggregate cardinality, as well as an approximation ratio of 1/4 under some conditions. Four other algorithms are presented and analyzed. We identify conditions under which some of these algorithms are optimal

%% file: Introduction11.tex
\section{Introduction and Related work}
\label{sec:intro}  
We consider a ``Social Group''  \cite{social_group} of networked nodes, where each node is looking for a common set of data files/segments (henceforth refereed to as the ``universe'') \cite{dist_selfish_repl,social_similarity}. The group has access to a common resource for downloading segments which is high cost, for example cost can be delay, license fee, transmission power etc. On the other hand, the nodes can exchange data among themselves at much lower cost possibly because of physical proximity, mutual self-help cooperation etc.. 
Initially, each node has a subset of the \emph{universe}, and it is interested in acquiring all other pieces of the universe through the low cost local area network rather than the high cost common resource to reduce the overall download cost. Various popular applications based on this concept include web caching, content distribution networks, peer-to-peer networks, etc., with the common theme of leveraging low-cost data exchange instead of consuming expensive resources. More recently, device-to-device (D2D) communication in cellular wireless networks is another example of such a social group, where mobile phones can either directly talk to other mobile phones without going through the base station or help other mobiles in their communication.

%

Principle assumption in resource sharing is cooperation and willing participation of all nodes. In particular, a node should be willing to let other nodes obtain data  from it. Typically, \emph{selfish nodes} are reluctant to provide data/help to others; on the other hand, they are keen to grab as much as possible from other nodes. In a peer-to-peer networking context, such selfish nodes are termed \emph{free riders}. In a scenario with a large social group of nodes, it is very likely that nodes do not ``know'' one another, and hence do not trust each other. This lack of trust acts as a deterrent to sharing, to the detriment of the whole community.


For rational nodes, free riding is the dominant strategy. However, if every node engages in free riding, there will be a complete breakdown of sharing. To discourage free riding, various policies and methods have been proposed in the context to peer-to-peer networks \cite{freewhite,overfreeride,cheapfreeriding,effort_based,global_fair,free_riding,g2g_free}. Policy counterparts of these can be used to counter the free riding problem in social groups also. Often, however, 
these policies suggest and incentivize good behavior, but do not enforce it. Selfish users manage to find weaknesses that can be exploited. Various types of attacks (such as whitewashing, collusion, fake services, Sybil attack \cite{free_riding,sybil,sybil_p2p}) are possible in P2P networks \cite{free_riding}. Similar attacks can be launched by nodes in social groups as well.

To prohibit free riding, we propose the \emph{Give-and-Take (GT) criterion} for segment exchange. 
Essentially, the idea is that a node $A$ can download segments from another node $B$ if and only if $A$ offers at least one new segment to $B$. Thus, a notion of \emph{fairness} is built into \gt --- $A$ cannot get anything from $B$ unless she offers $B$ something in return. Each side has to contribute at least one segment that the other side does not have. In this paper, we study the special case of ``Full Exchange,'' in which nodes exhibit altruistic behaviour --- $A$ is willing to provide \emph{all} segments it possesses even if she gets just \emph{one} segment from $B$. Because of this, after an exchange, both $A$ and $B$ will possess the \emph{union} of their individual segment sets before exchange. 

Incorporating \gt in our model, we study the problem of scheduling exchanges complying with \gt, so that the aggregate number of data segments available with all the nodes can be maximized. The motivation is 
\emph{reduced cost of downloading}---if the aggregate number of data segments acquired through low cost local exchanges is maximized, then the total number of segments to be downloaded using the high-cost resource is least.

Given data segment sets at each node, there may be several pairs of nodes that satisfy \gt and different choices for data exchange lead to different final aggregate cardinalities. Identifying the optimal schedule of data exchange that maximizes the aggregate cardinality of data segments available at all nodes is a combinatorial problem with exponential complexity ($\mathcal{NP}-$hard). The complexity of the problem can be accessed from a simple example illustrated in  Fig. \ref{fig:scheduling}, where there are only $4$ nodes and the universe size is $5$, and there are exponentially many data exchange schedules. Finding the one with optimal final aggregate cardinality is a challenging task.

This paper makes the following contributions.
\vspace{-0.05in} 
\begin{itemize}
\item We propose the novel {\it GT criterion}, which prohibits free riding in social groups. This, we believe, is a new concept that will help understand the fundamental principles of data sharing in local networks under fair exchange models.

\item We propose a Randomized algorithm, that works in phases, where in each phase it randomly picks a pair of nodes for exchange.  If the chosen pair satisfies \gt, then the exchange happens, otherwise the two nodes are kept aside and a new pair is chosen randomly again. The phase ends when all pairs are exhausted. We show that the randomized algorithm is optimal for maximizing the aggregate cardinality when the number of nodes is very large. We also show that the aggregate cardinality obtained by the randomized algorithm is at least $1/4$ times of the optimal aggregate cardinality under some conditions, thus providing a $4$ approximation.

\item We propose a Greedy-Links algorithm that, at each step, pairs those nodes so that the number of possible exchanges in the next step is maximized. We show that this algorithm is optimal for small number of nodes, e.g. $4$.

\item We propose shifted round-robin algorithm called the Polygon algorithm, that in each phase 
exchanges segments between neighboring nodes (certain predefined order of nodes), and then for the next phase, repeats the process after circularly left shifting the order of nodes.  The polygon algorithm is shown to be optimal when each node has at least one unique segment.

\item We propose two more algorithms, the first (called the Greedy-Incremental Algorithm) exchanges nodes that maximize the aggregate cardinality in next phase, while the second (called the Rarest First Algorithm) chooses that pair for exchange such that subset of segments with the minimum number of nodes is maximized. 

\item By means of extensive simulation results, we show that the Greedy links algorithm performs the best, since it tries to maximize the potential for nodes to exchange their pieces. Rarest first algorithm comes a close second to  the Greedy links algorithm, while the Randomized algorithm on an average outperforms Greedy incremental algorithm and Polygon algorithm. 
\end{itemize}


%% file: system_model10.tex
\section{System Model and Problem Formulation}
\label{sec:sysmodel}
We consider a set of nodes, $\mathcal{M}=\left\{1,2,\cdots, m\right\}$ and a universe $\mathcal{N}=\left\{1,2,\cdots, n\right\}$ of segments. Each node $i\in\mathcal{M}$ has an initial collection of segments $O_i\subset\mathcal{N}$. 

\emph{Give-and-Take (GT) criterion:} Two nodes $i, \ j\in \mathcal{M}$ with segment sets $X_i$ and $X_j$, respectively, can exchange segments 
if and only if $X_i\cap X_j^c \neq\emptyset \quad \text{and} \quad X_i^c\cap X_j \neq\emptyset$, i.e. node $j$ has at least one segment which is unavailable with node $i$ and vice versa. 
After exchange, both nodes have the segment set $X_i\cup X_j$. 



We consider the set of nodes $\mathcal{M}$ as the vertices in an undirected graph $G$, where an edge or {\it link} exists between two vertices $i,j\in\mathcal{M}$ if and only if they satisfy \gt. 
We denote the link between nodes $i$ and $j$ by the unordered 2-tuple $(i,j)$. By activating/pairing an edge/link we mean that the two nodes on that edge exchange their segments.

Given an initial collection of segment sets $O_i$, a schedule $S$ is a repeated activation of links (exchange of nodes) in graph $G$, in a given order. Note that since one link activation changes the graph $G$, we denote the dynamic graph $G$ as $G_{S}(r)$, where subscript $S$ stands for the schedule of node exchanges/link activations over graph $G$ and $r$ stands for the $r^{th}$ step of schedule $S$. Thus, $G_{S}(r)$ is a graph with vertex set $\mathcal{M}$, and where an edge between two vertices $i,j\in\mathcal{M}$ exists if they satisfy \gt at the $r^{th}$ iteration of schedule $S$.

For example, in $G_B(3)$ of Fig. \ref{fig:schedule_B}, i.e. in $3^{rd}$ step of schedule $B$, links exist between node $1$ and $4$, node $3$ and $4$ and node $2$ and $4$.  If link $(2,4)$ is activated, then node $2$ gets segment $5$ from node $4$, and node $4$ gets segments $2,3,4$ from node $2$. Thus, for $G_B(4)$ the graph is completely disconnected since no two nodes satisfy \gt.

We represent schedule $S$ as a row vector of links in the order of activation and $\left|S\right|$ denotes the number of link activations in the schedule $S$. We denote the set of available links after $r$ activations of a schedule $S$ by $\mathcal{L}(S,r)$. 

Let $O_i(S,r)$ denote the set of segments with node $i\in\mathcal{M}$ after the first $r$ activations of schedule $S$. In this context, $O_i(S,0)=O_i$ denotes the initial collection of segments with node $i\in\mathcal{M}$. Similarly, $\mathcal{L}([\cdot],0)$ denotes the set of links that exist between the initial segment sets.

A \emph{Maximal Schedule} is a schedule which cannot be extended by any other link activations, i.e. at the end of the maximal schedule \gt is not satisfied between any pair of nodes. Let $\mathcal{X}$ denote the set of all \emph{Maximal Schedules}.

For any schedule $S$, let $\left[S,\left(i,j\right)\right]$ denotes the schedule which follows all link activations in $S$, followed by activation of the link between node $i$ and $j$ such that $\left(i,j\right)\in\mathcal{L}(S,\left|S\right|)$.
\begin{figure}[h]
	\subfigure[Evolution of nodes' segment set under schedule $S_A$; Aggrgate cardinality under schedule $S_A$ is 20]{
		\label{fig:schedule_A}
		\includegraphics[width=\linewidth]{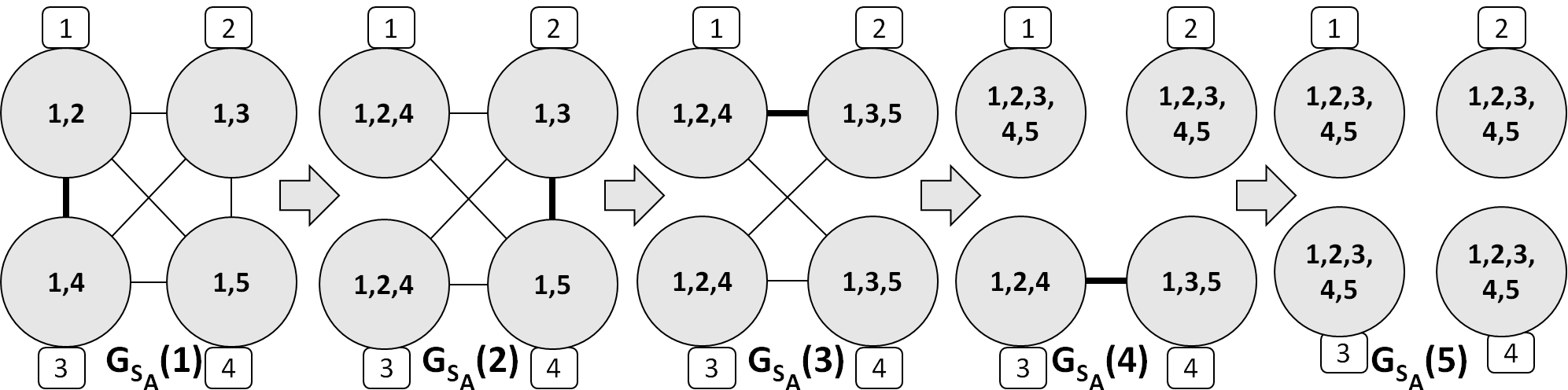}}
	\subfigure[Evolution of nodes' segment set under schedule $S_B$; Aggrgate cardinality under schedule $S_B$ is 17]{
		\label{fig:schedule_B}
		\includegraphics[width=\linewidth]{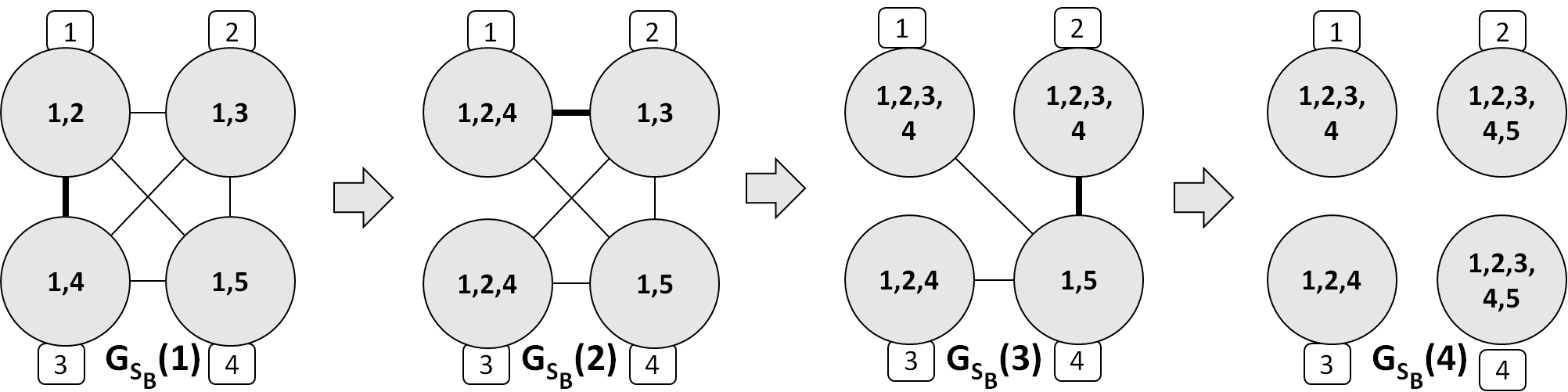}}
	\caption{Evolution of nodes' segment set under two different schedules. Nodes are numbered in the rounded squares. Node set is $\mathcal{M}=\{1,2,3,4\}$ and Universe is $\mathcal{N}=\{1,2,3,4,5\}$.}
	\label{fig:scheduling}
\end{figure}

In Fig. \ref{fig:schedule_B}, schedule $S_B$ is being followed for link activation. We activate link $(1,3)$ and both the nodes $1$ and $3$ will have segment sets $O_1(S_B,1)=O_3(S_B,1)=O_1\cup O_3$. We can activate either of links in $\mathcal{L}(S_B,1)=\left\{(1,2),(1,4),(2,3),(2,4),(3,4)\right\}$. Now, activating link $(1,2)$ followed by link $(2,4)$, leads us to graph $G_{S_B}(4)$ where no links exist between any of the nodes. Hence, $S_B$ is a maximal schedule which will be represented as $S_B=\left[(1,3),(1,2),(2,4)\right]$.

\subsection*{Assumptions}
\label{subsec:assumption}
\begin{enumerate}[{A}1:]
	\item We assume that the cost of exchange between nodes is zero or negligible as compared to the cost of downloading the segment from outside the set $\mathcal{M}$ similar to \cite{social_similarity}.
	\item $O_i\neq\emptyset, \text{ and }O_i\neq \mathcal{N}, \quad \forall\ i\in \mathcal{M}$.
\end{enumerate}
\subsection*{Computing the Social optimum}
\label{subsec:objective}
The aggregate cardinality on completion of a maximal schedule $S$ 
is denoted by $\alpha(S)=\sum_{i\in \mathcal{M}}\left|O_i(S,\left|S\right|)\right|$.
Our objective is to find a schedule of link activations $S^*$, such that
the total number of segments available with the nodes, \textit{i.e.,}
the aggregate cardinality on completion of $S^*$, is \emph{maximized}. Formally, we 
want to find
\begin{eqnarray}\label{eq:probform}
\max \quad \alpha(S). \\ \nonumber
\textrm{Subject to: } S \in \mathcal{X}
\end{eqnarray}

Note that a trivial upper bound to $\alpha(S) \le mn$ for even-$m$ and $\alpha(S) \le mn-1$ for odd-$m$, that will be useful later in the paper.

To understand the implications of solving (\ref{eq:probform}), note that to begin with, the number of ``missing'' pieces in the network is $(mn-\sum_{i\in\mathcal{M}} |O_i|)$. 
After the completion of optimal schedule $S^*$,
\emph{at most} $(mn-\alpha(S^*))$ segments will need to be fetched 
using the expensive resource, thus solving (\ref{eq:probform}) minimizes the social cost.\footnote{There is a possibility for nodes to download few segments over the high cost network after $S^*$ and again use the low cost network for exchanges among themselves. This process might be repeated recursively.}

The optimal schedule (solution)  to (\ref{eq:probform}) is the social optimum, however, solving 
(\ref{eq:probform}) has exponential complexity, that makes it computationally infeasible. Even for simple examples, as in Fig. \ref{fig:scheduling}, there are many possible schedules, and different schedules lead to different aggregate cardinalities, e.g.  
with schedule $S_A$ every node can get the universe while with schedule $S_B$ only two nodes can get the universe. In the rest of the paper, we propose several polynomial-time algorithms to solve (\ref{eq:probform}) and bound their performance analytically and through extensive simulations.

Next, we derive an important property of the \gt based segment exchanges in Lemma \ref{lem:at-least-2}.

\begin{lemma}
	\label{lem:at-least-2}
For any maximal schedule $S$ and $O_i\subsetneq \bigcup_{j\in\mathcal{M}} O_j\, \forall i, j\in\mathcal{M}$, at least two nodes will have $\bigcup_{i\in\mathcal{M}} O_i$ towards the end of schedule $S$.
\end{lemma}
\begin{proof}
	We reorder the nodes in ascending order of cardinalities at the end 
of $S$ and index them as $i_1, \dots, i_m$: $\left|O_{i_1}(S,\left|S\right|)\right|\leq\left|O_{i_2}(S,\left|S\right|)\right|\cdots\left|O_{i_{m-1}}(S,\left|S\right|)\right|\leq\left|O_{i_m}(S,\left|S\right|)\right|$.
	By \gt, if a link does not exist between the nodes $i$ and $j$, then $O_i\subseteq O_j$ or $O_j\subseteq O_i$. Therefore,
	$$O_{i_1}(S,\left|S\right|)\subseteq O_{i_2}(S,\left|S\right|)\cdots O_{i_{m-1}}(S,\left|S\right|)\subseteq O_{i_m}(S,\left|S\right|).$$
	Since the cardinalities of nodes $i_{m-1}$ and $i_m$ are the
largest and second largest,  and each link activation produces two nodes with same subset of segments, therefore $O_{i_{m-1}}(S,\left|S\right|)= O_{i_m}(S,\left|S\right|)$.
Since there is no link between $i_{m-1}, i_m$ and any other node, 
the segment sets of nodes $i_{m-1}$ and $i_m$ are supersets of 
segments available with all other nodes, therefore $$O_{i_{m-1}}(S,\left|S\right|)= O_{i_m}(S,\left|S\right|)=\bigcup_{i\in\mathcal{M}} O_i\left(S,\left|S\right|\right)=\bigcup_{i\in\mathcal{M}} O_i.$$ 
	Hence, at least two nodes will have $\cup_{i\in\mathcal{M}} O_i$ at 
the end of any maximal schedule $S$.
\end{proof}
\begin{cor}
	For any $S\in\mathcal{X}$, 
	$$\alpha(S)\geq 2\left|\bigcup_{i\in\mathcal{M}} O_i\right|+\sum_{i\in\mathcal{M}\setminus\left\{i_{m-1},i_m\right\}} \left|O_i\right|,$$
	where $i_{m-1}$ and $i_m$ are two nodes having the maximum cardinalities at the end of schedule $S$.
\end{cor}

Hence, Lemma \ref{lem:at-least-2}, shows that for any maximal schedule at least two nodes with get the whole {\it realized} universe, where by realized we mean the union of segments available with all nodes. In Lemma \ref{p_mnk}, we derive the probability that the realized universe is equal to the universe, when each node gets $k$-segments uniformly randomly drawn from universe ${\cal N}$.
\begin{figure*}
	\hrulefill\\
	$p^{m,n,k}$ is defined here for the sake of convenience; it is used in Lemma~\ref{p_mnk}.
	\begin{equation}
	\label{eq:pmnk}
	p^{m,n,k}={n\choose k}^{-m}\sum_{\begin{array}{c}
		\sum_{j=1}^{m-1} k_j = mk-n\\
		(k_1,k_2\cdots k_{m-1})
		\end{array}}{\prod_{j=1}^m {{(j-1)k-\sum_{i=1}^{j-2}k_i}\choose k_{j-1}}{{n-((j-1)k-\sum_{i=1}^{j-2}k_i)}\choose k-k_{j-1}}}. 
	\end{equation}
	\hrulefill
\end{figure*}
\begin{lemma}
	\label{p_mnk}
	Let each of the $m$ nodes select $k$ segments uniformly at random from a universe $\mathcal{N}$.  Then
	\begin{eqnarray}
	&&\mathbb{P}\left(\bigcup_{j=1}^m O_j = \mathcal{N}\right)=\begin{cases}
	p^{m,n,k} \quad \mathrm{if} \, mk\geq n, \\
	0 \quad\quad\quad \mathrm{otherwise.}
	\end{cases} \nonumber
	\end{eqnarray}
\end{lemma}
\begin{proof}
	Evidently, if $mk<n$, then
	$\mathbb{P}\left(\bigcup_{j\in \mathcal{M}} O_j= \mathcal{N}\right)=0$.
The interesting case is when $mk\geq n$, for which   
	$\mathbb{P}\left(\bigcup_{j\in \mathcal{M}} O_j= \mathcal{N}\right)=\dfrac{\text{No of favourable cases}}{\text{Total number of possibilities}}$.
	The number of $k$ element subsets of $\mathcal{N} ={n\choose k}$. Total number of possibilities$ = {n\choose k}^m$
	
	\emph{Calculating the number of favourable cases:} 
	If $mk > n$, some elements must appear in the segment sets of several nodes.
	The total number of repeated elements is given by $(mk-n)$. Let us assume	$\left|O_2 \bigcap O_1\right|= k_1$, $\left|O_3 \bigcap(O_1\bigcup O_2)\right|=k_2$, 
	$\dots$, $\left|O_m \bigcap \left(\bigcup_{i=1}^{m-1} O_i \right)\right| =k_{m-1}$.
	Therefore, $\sum_{j=1}^{m-1} k_j=mk-n$.
	The number of ways of choosing $O_1 = {n\choose k}$.
	$O_2$ has $k_1$ elements in common with $O_1$, so $k_1$ elements should be chosen from $O_1$ (consisting of $k$ elements) and the remaining $(k-k_1)$ elements should be chosen from the 
	remaining $(n-k)$ elements. 
	This gives the number of ways of choosing $O_2$ as ${k\choose k_1}{n-k\choose k-k_1}$.
	
	Similarly, $O_3$ has $k_2$ elements in common with $O_1$ and $O_2$, so $k_2$ elements should be chosen from $O_1 \cup O_2$ (consisting of $(2k-k_1)$ elements) and the remaining 
	$(k-k_2)$ elements from $(n-(2k-k_1))$ elements.
	Thus, the number of ways of choosing $O_3$ is ${2k-k_1\choose k_2}{n-(2k-k_1)\choose k-k_2}$.
	
	
	In the same way, the number of ways of choosing $O_j$ is
	$${{(j-1)k-\sum_{i=1}^{j-2}k_i}\choose k_{j-1}}{{n-((j-1)k-\sum_{i=1}^{j-2}k_i)}\choose k-k_{j-1}}.$$
	
	Hence, the total number of ways of choosing $O_1,O_2\cdots O_j\cdots O_m$  for a given $(k_1,k_2\cdots k_{j-1}\cdots k_{m-1})$ is 
	$$\prod_{j=1}^m {{(j-1)k-\sum_{i=1}^{j-2}k_i}\choose k_{j-1}}{{n-((j-1)k-\sum_{i=1}^{j-2}k_i)}\choose k-k_{j-1}}.$$
	
	To get the total number of favorable cases, we need to sum over all possible values of $(k_1,k_2\cdots k_{j-1}\cdots k_{m-1})$ for which $\sum_{j=1}^{m-1} k_j=(mk-n)$. Therefore,
	$$\mathbb{P}\left(\bigcup_{j=1}^m O_j = \mathcal{N}\right)=\begin{cases}
	p^{m,n,k} \quad \mathrm{if} \, mk\geq n, \\
	0 \quad\quad\quad \mathrm{Otherwise}.
	\end{cases} $$
\end{proof}
\begin{cor}
	The probability that at least one node will have the universe $\mathcal{N}$ at
	the end of any schedule $S\in\mathcal{X}$ is $p^{m,n,k}$.
\end{cor}
\begin{proof}
	This follows from Lemma~\ref{lem:at-least-2} and Lemma~\ref{p_mnk}. 
\end{proof}

%% file: heuristics15.tex
\section{Algorithms}
\label{sec:algorithms}
\subsection{Randomized Algorithm}
\label{subsec:rand_heur}
We first present a randomized algorithm to solve (\ref{eq:probform}).
The randomized algorithm works in phases. 
In each phase, two nodes are uniformly randomly chosen for exchange. If they have an edge between them in the corresponding graph, i.e. satisfy \gt, then they exchange their segments, otherwise they are kept aside without any exchange. Once a pair is chosen in a phase, with or without actual exchange, it takes no further part in that phase. The phase ends when choice of all pairs is exhausted, and then the new phase starts all over again, forgetting what happened in earlier phases. The algorithm terminates when there are no more links available in the graph after the end of any phase. 

\begin{algorithm}[h]
	\caption{Randomized Algorithm}
	\label{alg:rand_heur}
	\begin{algorithmic}[1]
		\State $r:=0, p:=1$ and $S_{rand}:=\left[\cdot\right]$
		\While {$\mathcal{L}(S_{rand},r)\neq \emptyset$}
			\State $\mathcal{M}_{pick}:=\mathcal{M}$
			\While {$\mathcal{M}_{pick}\neq \emptyset$} \Comment{phase $p$ begins}
			\State Select nodes $i$ and $j$ at random from $\mathcal{M}_{pick}$
			\If {$\left(i,j\right)\in \mathcal{L}(S_{rand},r)$}
				\State Activate link between nodes $i$ and $j$
				\State $S_{rand}\gets\left[S_{rand},(i,j)\right]$
				\State $r\gets r+1$
			\EndIf
			\State $\mathcal{M}_{pick}\gets \mathcal{M}_{pick}\backslash\left\{i,j\right\}$
			\EndWhile \Comment{phase $p$ ends}
			\State $p\gets p+1$
		\EndWhile
	\end{algorithmic}
\end{algorithm}

\begin{thm} If initial segment sets ($O_i$'s) are chosen uniformly at random from $\mathcal{N}$, with $\left|O_i\right|=k$, $1\leq k\leq n-1,\, \forall\, i\in\mathcal{M}$, 
\begin{enumerate}
\item Then the randomized algorithm is asymptotically optimal in $m$, i.e. for large $m$, $E(\alpha(S_{rand})) = nm$, where $nm$ is the upper bound on  (\ref{eq:probform}). 
\item The randomizeed algorithm is a $4$ approximation to (\ref{eq:probform}), if 
$k\in \left[\min\left(\dfrac{n}{\log_2 m},\dfrac{n}{4}\right),n-1\right]$. 
\end{enumerate}
\label{thm:rand_bound}
\end{thm}

\begin{proof}
We define a \emph{phase} as a round of link activations in steps 4-12 of Algorithm~\ref{alg:rand_heur}. For notational convenience, let $s^i_p = |O_{(i,p)}|$, where $O_{(i,p)}$ denotes the segment set with node $i$ at the beginning of phase $p$. At the start of phase 1, each node has $s^i_1=k$ segments, that are uniformly randomly picked from the universe. 
Since the algorithm is randomized, we focus on a particular node $i$ for analysis, and let in phase $p$, assume that node $i$ and $j_p$ are paired, where $j_p$ is uniformly randomly chosen among the other $m-1$ nodes. 

The segment set at $i^{th}$ node at the beginning of phase $p$ is $O_{(i,p)}=\left\{e_1,e_2,\cdots , e_{|O_{(i,p)}|}\right\}$. For every phase $p$, and any segment $e\in O_{(j_p,p)}$, we define the random variables $X_e(j_p,p)$:
$$X_e(j_p,p) =\begin{cases}
1\quad e \notin O_{(i,p)},\\
0\quad otherwise.
\end{cases} $$

Let $\mathcal{M}_0=\left\{i\right\}$ and $\mathcal{M}_p$ denote the subset of nodes that have been influenced by node $i$ by the end
of phase $p-1$. The set of influenced nodes is inductively defined to be all nodes that have had any pairings with node $i$ until phase $p-1$. For example, if node $i$ and node $2$ are paired in phase $1$, and nodes $2$ and $4$ are paired in phase $2$, then $\mathcal{M}_3 = \{i,2,4\}$.  Note that $\left|\mathcal{M}_p\right| \leq 2^{p-1}$ since in any phase only two nodes are paired with each other and therefore influence of $i$ grows only two-folds in each phase.
 

For phase $1$, node $i$'s segment set cardinality increase \emph{after exchange}, $(s_2^i-s_1^i)$,
will equal the number of segments that are available with node $j_1$ 
but not with node $i$, i.e.
\begin{eqnarray*}
	E\left(s_2^i-s_1^i | j_1 \notin \mathcal{M}_0 \right) &=& 
          E\left(\left.{\sum_{e\in O_{(j_1,1)}} X_e} \right| j_1 \notin \mathcal{M}_0 \right),\\
    &\stackrel{(a)}=&\left|O_{(j_1,1)}\right| \mathbb{P}\left(X_e=1\right), \\  
    & \stackrel{(b)}= & k\left(1-\dfrac{k}{n}\right).
\end{eqnarray*}
where $(a)$ follows from linearity of expectation and since $\mathcal{M}_0 = \{i\}$ implying that  $\{j_1 \notin \mathcal{M}_0\}$ is always true, and (b) follows from $\mathbb{P}\left(X_e=1\right) = (1 - \frac{k}{n})$ 
for every $e$. 
Moving on to phase $2$, similarly, $E\left(\left.{s_3^i-s_2^i}\right|j_2\not\in\mathcal{M}_1,\left|O_{(i,2)}\right|=c_2\right)$
\begin{align*}
= &E\left(\left.{\sum_{e\in O_{(j_2,2)}}X_e} \right| j_2\not\in\mathcal{M}_1, \left|O_{(i,2)}\right|=c_2
\right),\\
= &E (s^{j_2}_{2}) \left(1 - \frac{c_2}{n} \right).
\end{align*}
Then $E(\left.{s_3^i-s_2^i}\right|j_2\not\in\mathcal{M}_1),$
\begin{align}
	&=\sum_{c_2} E\left(\left.{s_3^i-s_2^i}\right|j_2\not\in\mathcal{M}_1,
	\left|O_{(i,2)}\right|=c_2\right)\mathbb{P}(\left|O_{(i,2)}\right|=c_2)\nonumber\\
	& \label{eq:int1}= E(s_2^{j_2})\left(1-\dfrac{E(s_2^i)}{n}\right),
\end{align}
where (\ref{eq:int1}) follows from linearity of expectation, and more importantly from the fact that for $j_2\not\in\mathcal{M}_1$, $O_{i,2}$ and $O_{j_2,2}$ are independent, since initially all segments at all nodes were drawn uniformly randomly and node $j_2$ has had no exchange with any node that had an exchange with node $i$, consequently, $s_2^i$ and $s_2^{j_2}$ are independent.
Also note that since the algorithm is randomized, $s_2^{j_2}$ and $s_2^i$ are identically distributed, hence $E(s_2^{j_2}) = E(s_2^i)$, and from (\ref{eq:int1})
\begin{equation}\label{eq:rand_eq}
E(\left.{s_3^i-s_2^i}\right|j_2\not\in\mathcal{M}_1)= E(s_2^i)\left(1-\dfrac{E(s_2^i)}{n}\right).
\end{equation}
Generalizing (\ref{eq:rand_eq}) for any phase $p$,
\begin{align}
\label{eq:incr}
E(\left.{s_{p+1}^i-s_{p}^i}\right|j_p\not\in\mathcal{M}_{p-1})=E(s_{p}^i)\left(1-\dfrac{E(s_{p}^i)}{n}\right),
\end{align}

Since $|{\cal M}_p| \le 2^{p-1}$, $\mathbb{P}\left(j_p\not\in\mathcal{M}_{p-1}\right)\geq \dfrac{\max\left(m-2^{p-1},0\right)}{m-1}$.
Therefore, from (\ref{eq:rand_eq})
\begin{align*}
	&E(s_{p+1}^i-s_p^i) \geq E(\left.{s_{p+1}^i-s_p^i}\right|j_p\not\in\mathcal{M}_{p-1})\mathbb{P}\left(j_p\not\in\mathcal{M}_{p-1}\right),\\
	& \geq E(s_p^i)\left(1-\dfrac{E(s_p^i)}{n}\right)\left(\dfrac{\max\left(m-2^{p-1},0\right)}{m-1}\right).
\end{align*}
Therefore, the expected cardinality of node $i$ at the end of phase $p$ is given by 
$E(s_{p+1}^i)$ \begin{equation}
\label{eq:low_rand}
\ge E(s_p^i) + E(s_p^i)\left(1-\dfrac{E(s_p^i)}{n}\right)\left(\dfrac{\max\left(m-2^{p-1},0\right)}{m-1}\right) 
\end{equation}
$\forall \ p\ge2$. Expected aggregate cardinality at the end of phase $p$ is given by $mE(s_{p+1}^i)$.

For any $m$, we can find a natural number $a$ such that $2^{a-1} < m \leq 2^a$. Then,
for iterations $p=2$ to $a-1$, $\left(\dfrac{\max\left(m-2^{p-1},0\right)}{m-1}\right) > \dfrac{1}{2}$.
Thus,  $E(s_{p+1}^i)$\begin{eqnarray}\nonumber
&\ge& E(s_p^i) + E(s_p^i)\left(1-\dfrac{E(s_p^i)}{n}\right)\left(\dfrac{\max\left(m-2^{p-1},0\right)}{m-1}\right), \\\label{eq:randspcase}
& \ge& E(s_p^i) + \frac{E(s_p^i)}{2}\left(1-\dfrac{E(s_p^i)}{n}\right), \ \ p=1,\dots, a-1.
\end{eqnarray}
Noting that $a - 1 = \lfloor \log_{2} m \rfloor$, we have $E(s^i_{\lfloor \log_{2} m \rfloor + 1})$
\[
 \geq E(s^i_{\lfloor \log_{2} m \rfloor}) + 
\frac{E(s^i_{\lfloor \log_{2} m) \rfloor}}{2}\left(1- \dfrac{E(s^i_{\lfloor \log_{2} m \rfloor})}{n}\right)
\]
Now the sequence $E(s^i_{\lfloor \log_{2} m \rfloor})$ is monotonically nondecreasing and 
bounded above by $n$; hence, the sequence converges to $n$.
Adding the cardinality of all nodes, we get $\alpha(S_{rand}) \ge nm$. 
Thus, asymptotically in $m$, i.e. for large number of nodes, the randomized algorithm achieves 
the optimal solution to (\ref{eq:probform}), 
since $nm$ is an upper bound on the aggregate cardinality of (\ref{eq:probform}). 
This proves part $(1)$ of the Theorem.

For part (2), we note that function $\frac{x}{2}(1-\frac{x}{n})$ is a concave function for $x=k, k+1, \dots, n$, with maximum at $x=n/2$, and increasing from $x=k,\dots,n/2$. Hence from (\ref{eq:randspcase}), either $E(s_{p+1}^i) > n/4$ or $E(s_p^i)\left(1-\dfrac{E(s_p^i)}{n}\right) > k$ for phases $p=1,\dots, a-1$. In the former case, adding the cardinality across $m$ nodes, we have $\alpha(S_{rand}) \ge \frac{mn}{4}$. In the latter case, $E(s_a^i) > \frac{ka}{2}\left(1-\frac{k}{n}\right)$ by adding for $a$ phases, and $\alpha(S_{rand}) \ge \frac{mka}{2}\left(1-\frac{k}{n}\right)$. We only consider $k < \frac{n}{2}$, since otherwise we already have the $4$-approximation, for which case $\left(1-\frac{k}{n}\right) > \frac{1}{2}$. Since $a =\log_2 m$, if $k\log_2 m > n$, then $\frac{mka}{2}\left(1-\frac{k}{n}\right) =\frac{mk\log_2 m}{2}\left(1-\frac{k}{n}\right) > \frac{mn}{4}$ as required. The proof for $k > \frac{n}{4}$ is immediate.
\end{proof}

\subsection{Greedy-Links Algorithm}
The Greedy-Links algorithm works by examining the impact of a choice $(i,j)$ on the \emph{number of links in the graph after} $i$ and $j$ are paired. At each decision point, the algorithm chooses that pair for which the number of links in the resulting graph is maximized. Ties are broken arbitrarily. The idea is motivated by the observation that if the number of links in the graph is large, then it is likely that the maximal schedule will be longer, and closer to an optimal one. The resulting maximal schedule is denoted by $S_{Glink}$. Algorithm~\ref{max_link} shows the pseudo-code. Time complexity for greedy-links algorithm is $\theta(m^6)$.
\begin{algorithm}[h]
	\caption{Greedy-Links Algorithm}
	\label{max_link}
	\begin{algorithmic}[1]
		\State $r:=0$ and $S_{Glink}:=\left[\cdot\right]$
		\While {$\mathcal{L}\left(S_{Glink},r\right)\neq\emptyset$}
		\State $w_{\left(i,j\right)}:=\left|\mathcal{L}\left(\left[S_{Glink},\left(i,j\right)\right],r+1\right)\right|\quad\forall \left(i,j\right)\in\mathcal{L}\left(S_{Glink},r\right)$
		\State $(i,j)=\underset{(i,j)\in\mathcal{L}(S_{Glink},r)}{\arg\max} w_{(i,j)}$
 \Comment{Ties are broken arbitrarily}
		\State Activate link between node $i$ and $j$
		\State $S_{Glink}\gets\left[S_{Glink},\left(i,j\right)\right]$
		\State $r\gets r+1$
		\EndWhile
		\end{algorithmic}
\end{algorithm}

We next show that Greedy Links algorithm can be shown to be optimal under certain cases. Even though the results stated next are for very restrictive scenarios, however, we conjecture that Greedy Links algorithm is close to optimal for all cases.
\begin{prop}
For $m=4$ nodes, if each node has a segment set $O_i$ 
consisting of exactly $k$ segments. Then, the Greedy Links algorithm is optimal.
\end{prop}
\begin{proof}
The proof follows by examining all possible equivalent graphs with $4$ nodes having identical number of $k$ segments, and then checking for optimality by brute force. We skip the details due to space constraints.
\end{proof}

\begin{prop}
 For $m=4$ nodes such that optimal aggregate cardinality is $4n$, i.e. all nodes can get all segments of the universe with an optimal algorithm.  Then the aggregate cardinality achieved by the \emph{Greedy links algorithm} is also $4n$, i.e., $\text{if } \alpha(S^*)=4n$ then $ \alpha(S_{Glink})=4n$.
\end{prop}
\begin{proof}
Without loss of generality, we only consider the case when none of the $4$ nodes have all the segments of the universe to begin with.

\emph{Claim 1:} For the optimal schedule $S^*$, $\left|\mathcal{L}(S^*,\ell)\right|\geq 4$, where for the first time after $\ell$ exchanges any two nodes can exchange to get all the segments of the universe.

Consider the optimal schedule $S^*$. From the hypothesis of the claim, the length of $S^*$ is $\ell+2$, $\left|S^*\right|=\ell+2$, since for the optimal schedule all nodes can get all segments of the universe and that can happen in two exchanges after $\ell$. For $S^*$, after the $\ell+1^{th}$ exchange, the number of links available is $\left|\mathcal{L}(S^*,\ell+1)\right|=1$, since two nodes have necessarily obtained all segments of the universe and cannot have links to any other nodes,  and the only link  present is between the remaining two nodes. Without any loss of generality, let us assume that $(1,2)$ link is still active after $\ell+1$ exchanges, $\mathcal{L}(S^*,\ell+1)=\left\{1,2\right\}$. 
	
Hence, link $(3,4)$ was activated in the $\ell+1^{th}$ exchange in $S^*$, such that $O_3(S^*,\ell)\cup O_4(S^*,\ell)=\mathcal{N}$. Also, $O_1(S^*,\ell)\cup O_2(S^*,\ell)=\mathcal{N}$ as $S^*$ ensures that all nodes get all the segments of the universe. Also note that $(\ell+1)^{th}$ activation yields first pair of nodes having universe.
	
Therefore, after $\ell$ exchanges, node $1$ has an edge to node $2$, and an edge to at least either of nodes $3$ or $4$, since $O_1(S^*,\ell) \subset \mathcal{N}, O_3(S^*,\ell) \subset \mathcal{N}, O_4(S^*,\ell) \subset \mathcal{N}$ and $O_3(S^*,\ell)\cup O_4(S^*,\ell)=\mathcal{N}$. Similarly, node $2$ will have an edge either to node $3$ or $4$. Therefore, $\left|\mathcal{L}(S^*,\ell)\right|\geq 4$.

\emph{Claim 2:} For Greedy links algorithm schedule $S_{Glink}$, if all nodes do not get all the segments of the universe, then $\left|\mathcal{L}(S_{Glink},\ell)\right|\le  3$, where for the first time after $\ell$ exchanges $S_{Glink}$ algorithm activates two nodes such that they get all the segments of the universe.

Subclaim 2.1: Under the hypothesis of Claim 2, the length of the $S_{Glink}$ schedule is $\ell+1$, i.e. $|S_{Glink}|=\ell+1$. 

The proof of Subclaim 2.1 is immediate since otherwise $S_{Glink}$ algorithm (by definition) would not have activated two nodes such that they get all the segments of the universe, and from the fact all nodes do not get all the segments of the universe.
%

At the $\ell+1^{th}$ exchange, let nodes $1$ and $2$ exchange their pieces. Since $|S_{Glink}|=\ell+1$, this will be the last possible exchange. Therefore, after the $\ell^{th}$ exchange, if node $3$ has edge to both node $1$ and $2$, then node $4$ cannot have an edge to either node $1$ or $2$, since otherwise the $S_{Glink}$ algorithm will not be at its last possible exchange. Hence, after the $\ell^{th}$ exchange, the set of edges are either $(1,2), (1,3), (1,4)$ or $(1,2), (2,3), (2,4)$, proving Claim 2.

The proof of the Theorem follows by a contrapositive argument to Claim 2, i.e. if we show that  
$\left|\mathcal{L}(S_{Glink}),\ell)\right|\ge  4$, where for the first time after $\ell$ exchanges $S_{Glink}$ algorithm activates two nodes such that they get all the segments of the universe, then with $S_{Glink}$ algorithm all nodes get all the segments of the universe. 

The claim that $\left|\mathcal{L}(S_{Glink}),\ell)\right|\ge  4$, follows from the fact the the optimal algorithm does so \emph{(Claim 1)} and the fact that Greedy links $S_{Glink}$ algorithm keeps maximum number of links \emph{alive}, and using brute force we can verify that  $S_{Glink}$ will also have at least $4$ edges when for the first time a pair of nodes exchange to get all the segments of the universe.
\end{proof}

\subsection{Polygon Algorithm}
The polygon algorithm is defined for a set of nodes $\mathcal{M}^U\subseteq\mathcal{M}$, 
such that each node $i\in\mathcal{M}^U$ has at least one unique segment with itself, i.e., 
$O_i\backslash \bigcup_{j\in\mathcal{M}^U\setminus\left\{i\right\}} O_j \neq \emptyset \quad \forall i\in \mathcal{M}^U$. 

Let us consider a row vector $P$ which is a permutation of nodes in 
$\mathcal{M}^U$. We define the \emph{left circular shift} of the row vector $P$:
$P=\left[i_1, i_2, i_3,\cdots, i_{\left|\mathcal{M}^U\right|}\right], 
LeftCircShift(P)=\left[i_2,i_3,i_4,\cdots, i_{\left|\mathcal{M}^U\right|},i_1\right].$

The proposed polygon algorithm in each round, picks two neighboring nodes in $P$ starting from  left, and activates the link between them and repeats this process until there are no pairs left. In the next round, the above procedure is repeated with $LeftCircShift(P)$. Note that this algorithm critically depends on the choice of starting permutation $P$. The starting permutation used in round $1$ is the one that has the \emph{smallest number of unique segments} placed at the right most location. 
For the case when, $\mathcal{M}^U \subset \mathcal{M}$, for maximizing the aggregate cardinality with the polygon algorithm,  we need to find the subset $\mathcal{M}^U$ with the largest cardinality. This, however, is a combinatorial problem and we use a greedy algorithm for this purpose that has linear complexity. The complexity of polygon algorithm is $\theta(m^3)$. Algorithm 3a provides the pseudo-code.

\begin{subalgorithms}
	\begin{algorithm}[h]
		\caption{Polygon algorithm}
		\label{alg:poly_specific}
		\begin{algorithmic}[1]
			\State Find $M^U$, Algorithm {\bf 3b},
			\While{$\left|\mathcal{M}^U\right|\geq 2$} 
			\State Let $M$ be any permutation of $\mathcal{M}^U$ \Comment{\gt is satisfied for any pair of nodes in $\mathcal{M}^U$}
			\State $l:=0, r:=0, S_{poly}:=\left[\cdot\right]$
			\While {$l \leq  \left\lfloor\dfrac{\left|\mathcal{M}^U\right|-1}{2}\right\rfloor$}			\For {$p=2:2:\left|\mathcal{M}^U\right|-1$}\Comment{$p$ is a even number}
			\State $i\gets (p-1)^{th}$ element in $M, j\gets p^{th}$ element in $M$
			\State Activating link between $i$ and $j$ 
			\State $S_{poly}\gets[S_{poly},\left(i,j\right)],\, r\gets r+1$
			\EndFor
			\State $M\gets LeftCircShift(M)$
			\State $l\gets l+1$
			\EndWhile
			\State Find $M^U$
			\EndWhile
		\end{algorithmic}
	\end{algorithm}
	\vspace{-10pt}
	\begin{algorithm}[h]
		\caption{Find $\mathcal{M}^U$}
		\label{alg:unique1}
		\begin{algorithmic}[1]
			\State $\mathcal{M}_{pick}:=\mathcal{M}$ and $\mathcal{M}^U:=\emptyset$
			\While {$\mathcal{M}_{pick}\neq\emptyset$}
			\State Choose a element $i$ from $\mathcal{M}_{pick}$
			\If {$O_i\setminus \cup_{j\in M^U} O_j\neq\emptyset$ and $\underset{j\in M^U}{\min}\left|O_j\setminus O_i\right| \neq 0$}
			\State $\mathcal{M}^U\gets\mathcal{M}^U\cup\left\{i\right\}$
			\EndIf
			\State $\mathcal{M}_{pick}\gets\mathcal{M}_{pick}\setminus \left\{i\right\}$
			\EndWhile
		\end{algorithmic}
	\end{algorithm}
	
\end{subalgorithms}
For the special case of $\mathcal{M}^U = \mathcal{M}$, i.e. each node has an unique segment, we show that the polygon algorithm is optimal for solving (\ref{eq:probform}).
\begin{prop}
	\label{poly_optimal}
	If $\mathcal{M}^U = \mathcal{M}$, then the polygon algorithm is optimal for solving (\ref{eq:probform}). 
	%
\end{prop}
\begin{proof} For lack of space, the proof is best described with a figure. In Fig. \ref{fig:polygon}, we consider $m=7$ nodes and each node $i$ contains at least one unique segment denoted by $i$. Each node can contain more segments, however, it is sufficient to just consider the unique segments. The figure is self-explanatory given the polygon algorithm, and in this odd-$m$ case achieves aggregate cardinality $mn-1$, where except one node every other node gets the universe, while for even-$m$ case every node gets the universe and achieves the aggregate cardinality of $mn$, that is maximum possible for (\ref{eq:probform}). If all segments of the universe are not seen by at least one node, then replace $n$ by $\left|\cup_{i\in {\cal M}} O_i\right|$, and the algorithm is still optimal.
	\begin{figure}[h]
		\centering
		\includegraphics[width=0.7\linewidth]{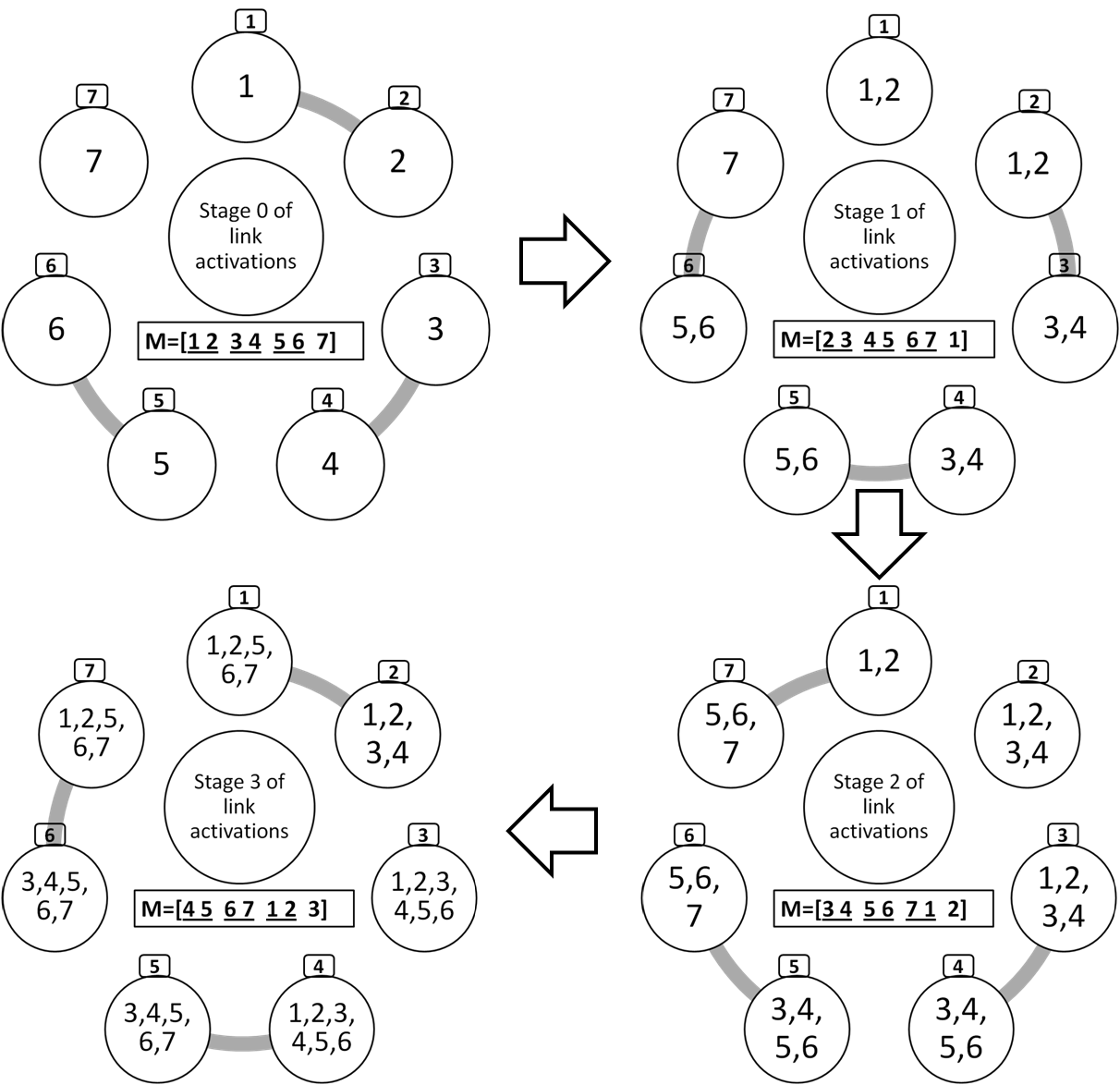}
		\caption{Demonstration of Polygon algorithm. }
		\label{fig:polygon}
	\end{figure}
\end{proof}

\subsection{Greedy-Incremental Algorithm}
The Greedy-Incremental algorithm works at each decision epoch, chooses that pair of nodes for exchange that has maximum sum of newly added segments to each node after the exchange. To each link $(i,j) \in \mathcal{L}\left(S_{Ginc},r\right)$, we assign a weight $w(i,j)$ (in step 3 of Algorithm~\ref{max_weight}) given by the increment of aggregate cardinality if $(i,j)$ is activated. The link with the maximum weight is selected for activation. Ties are broken randomly. Greedy-Incremental algorithm has complexity of $\theta(m^4)$.

\begin{algorithm}[h]
	\caption{Greedy-Incremental Algorithm}
	\label{max_weight}
	\begin{algorithmic}[1]
		\State $r:=0$ and $S_{Ginc}:=\left[\cdot\right]$
		\While {$\mathcal{L}(S_{Ginc},r)\neq \emptyset$}
		\State $w_{(i,j)}:=2\left|O_i(S_{Ginc},r)\cup O_j(S_{Ginc},r)\right|-\left|O_i(S_{Ginc},r)\right|-\left|O_j(S_{Ginc},r)\right|\quad \forall (i,j)\in\mathcal{L}(S_{Ginc},r)$
		\State $(i,j)=\underset{(i,j)\in\mathcal{L}(S_{Ginc},r)}{\arg\max} w_{(i,j)}$ 
		\Comment{Ties are broken arbitrarily}
		\State Activate link between node $i$ and $j$
		\State $S_{Ginc}\gets\left[S_{Ginc},\left(i,j\right)\right]$
		\State $r\gets r+1$
		\EndWhile
	\end{algorithmic}
\end{algorithm}
\vspace{-10pt}
\subsection{Rarest First Algorithm}
Rarest first algorithm activates the link so that availability for the subset of segments with the minimum number of nodes is maximally increased. The steps have been described in detail in Algorithm \ref{rare_first}. Rarest first algorithm has complexity of $\theta(nm^4)$.
\begin{algorithm*}
	\caption{Rarest first Algorithm}
	\label{rare_first}
	\begin{algorithmic}[1]
		\State $r:=0, S_{rare}:=\left[\cdot\right]$
		\While {$\mathcal{L}(S_{rare},r)\neq\emptyset$}
		\State $\mathcal{N}_p:=\left\{e\in \mathcal{N}: \underset{j\in \mathcal{M}}{\sum}\mathbb{I}_{\left\{e\in O_j(S_{rare},r)\right\}}=p\right\} \quad\forall 1\leq p\leq m$
		\State $f(\mathcal{N}_p,(i,j)):=\sum_{e\in \mathcal{N}_p} \mathbb{I}_{\left\{e\in O_i(S_{rare},r)\cap O_j^c(S_{rare},r)\right\}}+\mathbb{I}_{\left\{e\in O_j(S_{rare},r)\cap O_i^c(S_{rare},r)\right\}} \quad\forall 1\leq p\leq m, (i,j)\in \mathcal{L}(S_{rare},r)$ \Comment{Increment in the availaibility of segments in $\mathcal{N}_p$ if link between $i$ and $j$ is activated}
		\State $\mathbf{R}_{(i,j)}=\left[\mathbb{I}_{\left\{O_i(S_{rare},r)\cup O_i(S_{rare},r)\neq \mathcal{N} \right\}}\quad f(\mathcal{N}_1,(i,j))\quad f(\mathcal{N}_2,(i,j))\quad\cdots\quad f(\mathcal{N}_m,(i,j))\right]_{1\times (m+1)} \forall (p_1,p_2)\in \mathcal{L}(S_{rare},r)$
		\State Define link activation preference matrix $\mathbf{L}_A$ with $\mathbf{R}_{(i,j)} \forall (i,j)\in \mathcal{L}$ as row vectors.
		\State Sort rows of $\mathbf{L}_A$ such that column 1 is descending order, then by column 2 in descending order, then by column 3 in descending order and so on.
		\State $(i,j):=\arg (\text{Row at the top of }\mathbf{L}_A)$
		\EndWhile
	\end{algorithmic}
\end{algorithm*}

In each round (Lines 2-9), universe $\mathcal{N}$ is divided in $m$ partitions such that $p^{th}$ partition $\mathcal{N}_p$ of segments is available with $p$ nodes. The function $f\left(\mathcal{N}_p,\left(i,j\right)\right)$ evaluates the increment in the availability of segments in $\mathcal{N}_p$ if the link between nodes $i$ and $j$ are allowed to exchange segments where $1\leq p\leq m$ and $\left(i,j\right)\in\mathcal{L}\left(S_{rare},r\right)$. Then for each of the links $\left(i,j\right)\in\mathcal{L}\left(S_{rare},r\right)$ we define a row vector $\mathbf{R}_{\left(i,j\right)}$ having $m+1$ entries. First entry indicates the whether the nodes will have the universe if link $\left(i,j\right)$ is activated. Following entries record the increment in each of the partitions of the universe if the link $\left(i,j\right)$ is activated in increasing order of $p$.

Then we arrange different rows in form of a link activation matrix $\mathbb{L}_A$ such that column 1 is in descending order, followed by column 2 and so on. Sorting rows of $\mathbb{L}_A$ in descending order by column 1 deters activating links which can produce nodes having universe. To break the tie among links having same value in column 1, we sort in descending order by column 2. This would give more preference to links activating whom will increase the segments which are available with only 1 node or are rarest. To break the tie further we use column 3 so that availability of the segments which are with 2 nodes can be increased. This process is continued for all columns.

After completion of this process link corresponding to the top row is activated.

%% file: results9.tex
\vspace{-7pt}
\section{Performance Comparison}
\label{sec:results}
Fig.~\ref{fig:all} shows a performance comparison of the various algorithms. For a particular set of (number of nodes $m$, universe size $n$, initial segment set size $k$) values, we generate 100 segment sets uniformly at random. For each such ``sample point'' or ``run'' the optimal aggregate cardinality is computed, and each of the five algorithms simulated. The average aggregate cardinality values (over 100 sample points) are shown, along with 95\% confidence intervals.
\begin{figure*}[ht]
  \centering
  \includegraphics[width=0.68\textwidth]{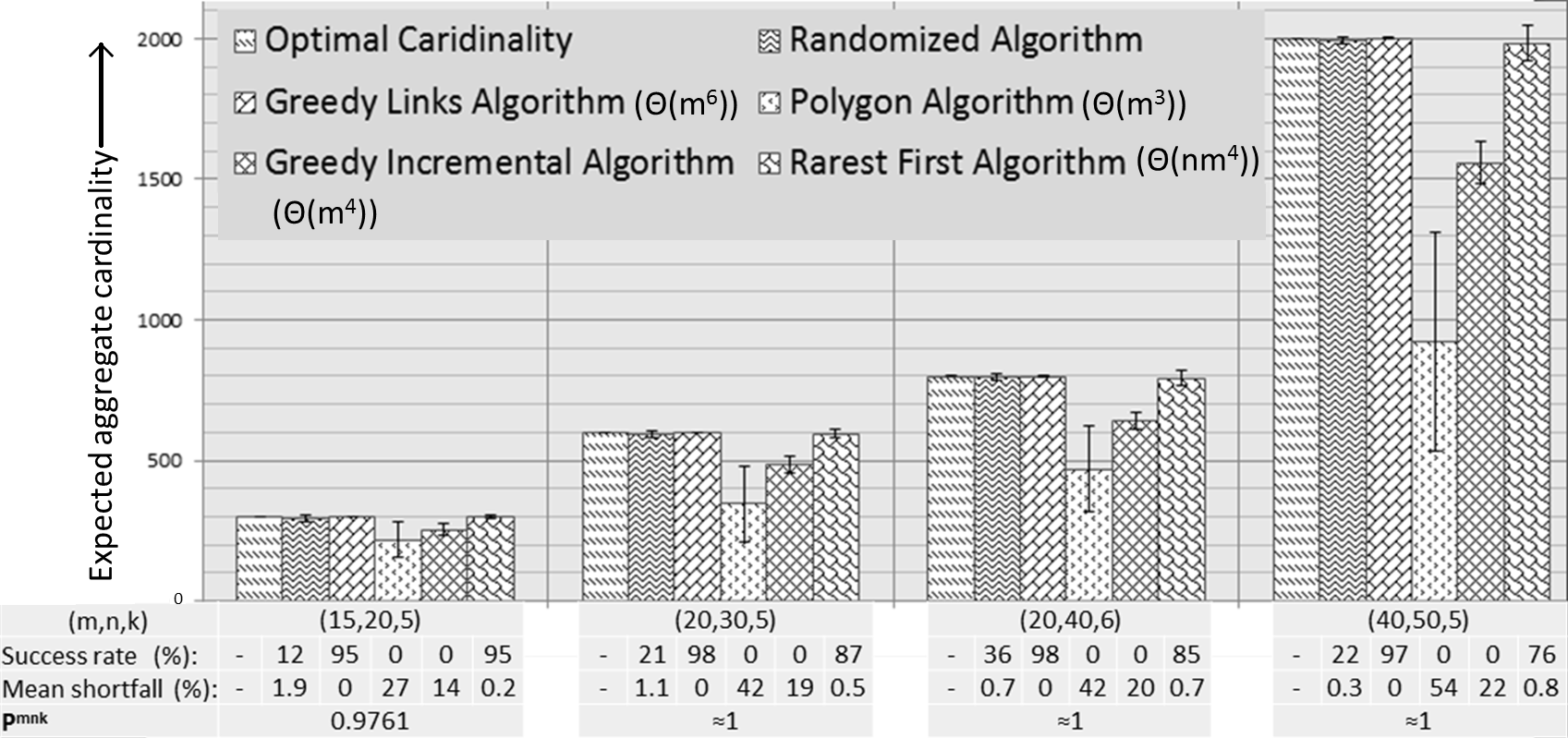}
  \caption{Expected aggregate cardinality achieved by different algorithms, averaged over 100 randomly generated segment sets. Four different problem instances ($(m,n,k)$) are shown, with corresponding $p_{m,n,k}$ values. For each object set, the success rate and mean shortfall of each algorithm are shown as well.}
  \label{fig:all}
\end{figure*}

If an algorithm achieves the \emph{optimal} aggregate cardinality corresponding to a
sample point, the algorithm is said to be ``successful'' on that run. 
Fig.~\ref{fig:all} also shows the ``success rate'' of an algorithm, which is the
fraction of runs for which the algorithm was successful. The average shortfall from
the optimal aggregate cardinality is shown as well. Lastly, Fig.~\ref{fig:all} 
displays the value of $p_{m,n,k}$, calculated using \eqref{eq:pmnk}.
We recall that $p_{m,n,k}$ is the probability
that after the initial choice of segment sets uniformly at random, the universe
$\mathcal{N}$ \emph{is} available among the nodes. 

Our results show that the Greedy Links ($GL$) algorithm performs the best---it is able to
achieve an aggregate cardinality close to the optimal consistently (success rate
$\geq 95\%$). It is followed closely by the Rarest First ($RF$) 
with similar success rates and mean shortfalls. On an average, Randomized algorithm ($R$) also performs better than Greedy Incremental ($GI$) and Polygon ($P$) algorithms.

The aggregate cardinality achieved by the $R$ algorithm is \emph{very}
close to the optimal, but in most runs, it falls just a little bit short.
For this reason, we see in Fig.~\ref{fig:all} that the heights of the bars corresponding
to the optimal and $R$ algorithms match almost exactly, yet the success rate of
the $R$ algorithm is quite low. On the other hand, the $GL$ algorithm enjoys a very
high success rate, and therefore, is nearly optimal. It is can be seen that the
mean shortfall corresponding to the $GL$ algorithm is zero.

The $RF$ algorithm starts off on a very promising note, but as the problem size
gets bigger, there is a dip in performance. As $(m,n,k)$ increases from (15,20,5)
to (40,50,5), the success rate drops from 95\% to 76\%, while the mean shortfall
increases from 0.2\% to 0.8\%. It is noteworthy that the shortfall of the $R$ algorithm 
reduces as the problem size gets bigger.

Our results indicate that the $GI$ algorithm does not perform well in
general, even though it is a ``natural'' greedy algorithm---it chooses to activate
the link that yields the largest immediate benefit (the largest increase in
aggregate cardinality). Similarly, even though the $P$ algorithm can be
optimal under certain conditions, it does not do well in general.

The significance of $p_{m,n,k}$ is this: If $p_{m,n,k}$ is high, then there is a good
chance that the social group will need to download \emph{only a few} segments using
the expensive resource. A high value of $p_{m,n,k}$ means that with high probability,
the full universe $\mathcal{N}$ is scattered among the nodes in the group. In such a
situation, our results suggest that the GL or RF or R algorithm is very likely
to achieve near-optimal aggregate cardinality, leaving only a few segments to 
be downloaded.

\textit{Remark:} We note that the more complex algorithms exhibit better performance
than the less complex ones.


In Table~\ref{tab:bounds}, we focus on the Randomized algorithm, and examine the
expected aggregate cardinality achieved for different sets of $(m,n,k)$. It can be
seen that as the number of nodes $m$ increases, the expected aggregate cardinality
approaches $mn$; this is consistent with the asymptotic optimality 
and approximation ratio results shown in Theorem~1.

\begin{table}[h]
	\centering
	\caption{Bounds for randomized algorithm}
	\begin{tabular}{|c|c|c|c|c|}
		\hline
		m & n & k & $\begin{array}{c}
		E(\alpha(S_{rand}))\\
		\text{(Simulation)}
		\end{array}$ & \parbox[c]{0.15\textwidth}{Lower bound on mean aggregate cardinality} \\
		\hline
		60 & 100 & 3 & $5027.0\pm 347.9$ & 3867.4 \\
		\hline
		60 & 100 & 5 & $5715.7\pm 219.1$ & 4829.2 \\
		\hline
		60 & 100 & 7 & $5919.4\pm 127.6$ & 5318 \\
		\hline
		80 & 200 & 15 & $15959\pm 102$ & 15106 \\
		\hline
		100 & 300 & 15 & $29819\pm 254$  & 27486 \\
		\hline
	\end{tabular}%
	\label{tab:bounds}%
\end{table}%

%% file: conclusions10.tex
\vspace{-10pt}
\section{Conclusions and Future Work}
\label{sec:concl}
The problems studied in this paper were motivated by the context of a social group, in which each group member was interested in obtaining a universe of segments. The approach was to study how mutual exchanges among the group members can help in maximizing the aggregate cardinality of nodes' segment sets, at low cost. To tackle the problem of free riding that arises in such situations, we proposed the novel GT criterion, and explored a number of algorithms for exchange of segments, where each exchange had to be GT-compliant. Analysis of the algorithms yielded interesting properties, and performance was benchmarked against the optimal aggregate cardinality.

The present model assumes that nodes arrive empty-handed and download segments using the expensive resource, to kick-start the process of GT-compliant mutual exchanges. After this process is over, still more segments may need to be downloaded; it will be interesting to understand the best download strategy in the second round, utilizing the state of the system at the end of the first. Further, a central scheduler with a view of all the 
nodes' segment sets has been assumed in this paper; we would like to understand the issues when nodes act on their own.